\documentclass[a4paper, USenglish]{lipics-v2016}

\usepackage{etex}
\usepackage{amsmath}
\usepackage{amssymb}
\usepackage{amsthm}
\usepackage{mathabx}
\usepackage{graphicx}
\usepackage{wrapfig}
\usepackage{psfrag}
\usepackage{color}
\usepackage[ruled,vlined,linesnumbered,titlenotnumbered]{algorithm2e}

\SetKwProg{Fn}{Function}{ is}{end}
\SetKwBlock{Do}{Do forever :}{end}
\SetKwFor{Upon}{Upon receipt of}{:}{endfor}
\SetKwFor{Forall}{forall}{}{endfor}

\usepackage{tikz}
\usepackage{tkz-graph}
\usetikzlibrary{decorations.pathreplacing}
\usetikzlibrary{plotmarks,trees,arrows}
\usetikzlibrary{calc}
\usetikzlibrary{graphs}
\usetikzlibrary{shapes.geometric}
\usepackage{pgf}
\usepackage{comment}

\usepackage{mathtools}

\newcommand{\ID}{\mbox{\sc id}}
\newcommand{\port}{\mbox{\sf pt}}
\newcommand{\ports}{\mbox{\sf Ports}}
\newcommand{\clr}{\mbox{\sf c}}

\newcommand{\cnt}{\mbox{\sf cnt}}
\newcommand{\wait}{\mbox{\sf wait}}
\newcommand{\ord}{\mbox{\sf Ord}}
\newcommand{\tmp}{\mbox{\sf tmp}}

\newcommand{\Msgs}{\mbox{$\cal M$}}
\newcommand{\Alg}{\mbox{$\cal A$}}

\newcommand{\Bit}{\mbox{\rm Bit}}

\begin{document}
%%%%%%%%%%%%%%%%%%%%%%%%%%%%%%%%%%%%

\title{Resource Efficient Stabilization for Local Tasks despite Unknown Capacity Links}%Compact Tidying and Coloring despite Adversarial Garbage}

%\begin{CCSXML}
%<ccs2012>
%<concept>
%<concept_id>10003752.10003809.10010172</concept_id>
%<concept_desc>Theory of computation~Distributed algorithms</concept_desc>
%<concept_significance>500</concept_significance>
%</concept>
%</ccs2012>
%\end{CCSXML}

%\ccsdesc[500]{Theory of computation~Distributed algorithms}


\author[1]{L\'elia Blin}
\author[2]{Ana\"{\i}s Durand}
\author[3]{S\'ebastien Tixeuil}
\affil[1]{Sorbonne Universit\'e, UPMC Univ Paris 06, CNRS, Universit\'e d'Evry-Val-d'Essonne, LIP6 UMR 7606, 4 place Jussieu 75005, Paris, France.%\\
%\texttt{lelia.blin@lip6.fr}
 }
 \affil[2]{Université Clermont Auvergne, LIMOS, Clermont-Ferrand, France.}
\affil[3]{Sorbonne Universit\'e, CNRS, LIP6, FR-75005, Paris, France.%\\
 % \texttt{firstname.lastname@lip6.fr}
 }
\maketitle

% !TEX root = Color_msg.tex
 \begin{abstract}
Self-stabilizing protocols enable distributed systems to recover correct behavior starting from any arbitrary configuration. In particular, when processors communicate by message passing, fake messages may be placed in communication links by an adversary. When the number of such fake messages is unknown, self-stabilization may require huge resources:
\begin{itemize}
\item generic solutions (\emph{a.k.a.} data link protocols) require unbounded resources, which makes them unrealistic to deploy,
\item specific solutions (\emph{e.g.}, census or tree construction) require $O(n\log n)$ or $O(\Delta \log n)$ bits of memory per node, where $n$ denotes the network size and $\Delta$ its maximum degree, which may prevent scalability.
\end{itemize}

We investigate the possibility of resource efficient self-stabilizing protocols in this context. Specifically, we present a self-stabilizing protocol for $(\Delta+1)$-coloring in any $n$-node graph, under the asynchronous message-passing model. The problem of $(\Delta+1)$-coloring is considered a benchmarking problem for local tasks. Our protocol offers many desirable features.

It is deterministic, it converges in $O(k\Delta n^{2}\log n)$ message exchanges, where $k$ is the bound  of the link capacity in terms of number of messages, and it uses messages on $O(\log \log n + \log \Delta)$ bits with a memory of $O(\Delta \log \Delta+\log \log n)$ bits at each node. 
The resource consumption of our protocol is thus almost oblivious to the number of nodes, enabling scalability.

Moreover, a striking property of our protocol is that the nodes do not need to know the number, or any bound on the number of messages initially present in each communication link of the initial (potentially corrupted) network configuration. This permits our protocol to handle any future network with unknown message capacity communication links.

A key building block of our coloring scheme is a spanning directed acyclic graph construction, that is of independent interest, and can serve as a useful tool for solving other tasks in this challenging setting.
\end{abstract}

\newpage
\section{Introduction}

Self-stabilization~\cite{D74j,D00b,T09bc} is a versatile technique that enables recovery after arbitrary \emph{transient} faults hit the distributed system, where both the participating processes and the communication medium are subject to be corrupted. Roughly, a self-stabilizing protocol is able to bring the system back to a legal configuration, starting from an arbitrary initial, potentially corrupted, configuration.  The core motivation for designing self-stabilizing protocols has been underlined by Varghese and Jarayam~\cite{VJ00j}, who observed that, whenever processes can crash and recover, a message-passing distributed system may reach any arbitrary global state, where the local variables stored at the processes may be inconsistent, and/or the communication links may contain spurious erroneous messages. As the global state is arbitrary, one may even assume that the local variables at the nodes as well as the contents of the messages are adversarially set to prevent recovery. 

It is worth pointing out that self-stabilization in the presence of fake messages in the communication links is a challenge. Specifically, it is particularly difficult to ensure recovery when no upper bound is known on the number of (possibly fake) messages in transit in the initial configuration. Conversely, if an upper bound is known, then one can reset the system to a clean configuration by, first, emptying the links, and, second, resetting all local variables using a protocol that can ``trust'' the messages. Hence, non surprisingly, the vast majority of recent works in self-stabilization assumes a weaker adversary than the one we consider in this paper. In particular, a widely used model of self-stabilization is the \emph{state model}. 

\subparagraph{State Model, and Data Link Protocols.}

In the state model, processes atomically read the states of their neighboring processes for updating their state. That is, the state model abstracts away all issues related to corrupted communication media. Indeed, the state model is motivated by the fact that there exist self-stabilizing \emph{data link} protocols~\cite{AB93j,BGM93j,DDPT11j,DIM97j,GM91j,KP93j,V00j}. The purpose of such protocols is to ensure reliable communications between neighbors exchanging messages over unreliable communication links. Yet, the use of data link protocols such as the ones in the aforementioned previous work yields  important issues. In particular, if the initial number of spurious messages in the communication links is unknown, then it is  impossible to design a self-stabilizing data link protocol using bounded memory at each node~\cite{AB93j,BGM93j,DIM97j,GM91j,KP93j}. Therefore, one has to relax the constraints on the protocol. Such a relaxation may consist in designing \emph{pseudo-stabilizing} data link protocols~\cite{BGM93j}. However, a pseudo-stabilizing protocol only guarantees that an infinite suffix of the execution satisfies the specification of the system, and hence its stabilization time becomes unbounded. Another relaxation consists of using self-stabilizing data link protocols that are \emph{not bounded} in terms of resources. However, such protocols require either unbounded variables~\cite{DIM97j,GM91j,KP93j}, or unbounded code size (\emph{a.k.a.}\! aperiodic functions)~\cite{AB93j}, which is undesirable from a practical point of view. A third relaxation consists of using \emph{randomization}~\cite{AB93j}, but then the correctness of the system is not certain. 

As a consequence, in the framework of data link self-stabilization, previous work often assume that the initial number of spurious messages present in the communication links is known to the participating processes. Under this assumption, very efficient self-stabilizing solutions can be obtained (see, \emph{e.g.},~\cite{V00j}). Actually, assuming that the number of erroneous messages initially present in the links is known to the nodes, even stronger self-stabilizing properties can be guaranteed, such as \emph{snap-stabilization}~\cite{DDPT11j} (a snap-stabilizing data link protocol guarantees that reliable communications between participating processes are immediately available after a failure). 

From the above, one can conclude that, while the  knowledge of the number of possible initial spurious messages in the communication links enables the design of efficient self-stabilizing data link protocols, the lack of this knowledge precludes the existence of bounded-size self-stabilizing deterministic data link solutions. In other word, when no bound is known, the use of data link protocols does not provide fully practical solutions for the design of efficient self-stabilizing protocols. Therefore, one has to focus on self-stabilization for message-passing systems, without using data links.

\subparagraph{Message-Passing Model.}

Self-stabilizing protocols that operate in message passing systems with unknown initial link capacity are the most versatile, since they can directly be executed in newly set up networks whose characteristics were unknown when the protocol was designed. However, previous works that do not rely on a data link layer require large messages, large memory, or both. For instance, the versatile \emph{census protocol} in~\cite{DT02j} collects the entire topology at each node, and therefore requires message and memory of $O(n \Delta \log n)$ bits in $n$-node networks with maximum degree~$\Delta$. Similarly, the versatile technique based on so-called \emph{$r$-operators}, where the algebraic properties of the executed protocol  guarantees self-stabilization, has been shown to be quite efficient in general~\cite{DDT06j,DT01jb,DT03j}. However, to our knowledge, in the message passing setting, there exists $r$-operators only for variants of tree construction tasks~\cite{DDT06j}. Moreover, while the size of the messages used by such protocols remain in $O(\log n)$ bits, the memory at each node might grow as much as $\Omega(\Delta \log n)$ bits. 

\subparagraph{Local tasks. Vertex coloring}

One may wonder whether the large amount of memory and communications resources used in the context of unknown capacity links is due to the global nature of the task to be solved (census, tree construction). Hence, an intriguing question arises: do local tasks yield high resource consumption when solved self-stabilizingly with unknown capacity links?
A benchmarking local task in the domain of self-stabilization is that of vertex coloring.
In vertex coloring, every process in the network must maintain a color variable such that, for every two adjacent nodes, the value of their color variables is distinct. Typically, the number of colors is supposed to be restricted in the range $\{1,2,\dots,\Delta+1\}$ in networks with nodes with maximum degree $\Delta$. Vertex coloring is one of the most studied tasks in distributed network computing in general, and in self-stabilization in particular, as witnessed by numerous contributions: \cite{BEG17r,BDPPT10c,BDGT09c,BT18c,GK93j,GT00c,HT04c,MT09j,MSTFG11j,SS93j}. While most previous work about self-stabilizing vertex coloring considered the state model (see~\cite{BDGT09c,BT18c,GK93j,GT00c,SS93j}), a few paper considered the message passing model~\cite{BEG17r,HT04c,MT09j,MSTFG11j}. However, all existing solutions for this latter model provide probabilistic guarantees only~\cite{BEG17r,HT04c,MT09j,MSTFG11j}, and most of them assume some strong or weak forms of synchronous execution model~\cite{BEG17r,HT04c,MSTFG11j}. To our knowledge, there are no deterministic self-stabilizing vertex coloring protocols that operate in the \emph{asynchronous} message passing setting, where the number of initial spurious messages in communication links is \emph{unknown} to the participating processes. 

\subparagraph{Our contribution.}
To this paper, it remains unknown whether resource efficient self-stabilizing solutions exist when communication links have unknown initial capacity. 
We show that, for some local tasks, the answer to this question is positive. 

In more details, we establish the following result: without any assumption on the number of messages present initially (and their content), starting from any configuration, we present a protocol that is self-stabilizing for the task of vertex coloring and uses $O(\log \log n + \Delta \log \Delta)$ bits of memory per node, and $O(\log \log n + \log \Delta)$ bits of information per message, where $n$ denotes the number of nodes in the network, and $\Delta$ its maximum degree. 
A key ingredient of our protocol of independent interest is a symmetry breaking mechanism that locally orients every link in the network so that the overall orientation is acyclic (hence constructing a directed acyclic graph in the network), simplifying the design of higher layer algorithms such as vertex coloring or maximal independent set. 

Our work thus paves the way toward resource efficient self-stabilizing protocols for the most challenging communication model, enabling solutions to remain valid when new networks are considered.

\section{Model}

The communication model consists of a point-to-point communication network described by a connected graph $G = (V, E)$ where the nodes $V$ represent the processes and the set $E$ represent bidirectional communication channels. Processes communicate by message passing: a process sends a message to another by depositing the message in the corresponding channel. We denote by $N(v)$ the set of processes that are \emph{neighbors} with process $v$, \emph{i.e.}, such that there exists a communication link between them and $v$.
Let denote by $n$ the number of processes and $\Delta$ the degree of the graph. Note that we denote by $\delta(v)$ the degree of node $v$.

\subparagraph{Communications.}

The communication model is asynchronous message passing with FIFO channels (on each link messages are delivered in the same order as they have been sent).  The number of messages per link is bounded by an integer $ k $, however the nodes do not know $k$. $Q(u,v) = (m_q,m_{q-1}, \dots, m_1)$ is the queue representing the messages in FIFO order between $u$ and $v$, where $m_1$ is the head of the queue and $q \leq k$. We assume each node $v$ is \emph{fair} with respect to its input channels: if $v$ receives an infinite number of messages, then a given message $m$ cannot stay in $v$'s input channel forever. We denote by \Msgs\ the set of all possible messages.
A node $v$ has access to locally unique port numbers associated with its adjacent edges. We do not assume any consistency between port numbers of a given edge. In short, port numbers are constant throughout the execution but two neighboring processes can associate different port numbers for the communication link between them. We denote by $\ports(v)$ the set of port numbers for the adjacent edges of process $v$. The port number associated by $v$ to the edge $(v,u)$, if it exists, is denoted $\port(v,u)$.

\subparagraph{Execution.} Each process $v$ maintains some \emph{variables}. We denote by $var_v$ the value of variable $var$ at process $v$. The \emph{state} of a process is the vector of the values of its variables. A \emph{configuration} is the vector of the state of every process and the content of the channels between every two neighboring processes. We denote by $var_v(\gamma)$ the value of variable $var_v$ in configuration $\gamma$. Similarly, $var_v[u,\gamma]$ and $Q(\gamma,v,u)$ denote the value of the entry $u$ of array $var_v$ in $\gamma$ and the content of channel $Q(v,u)$ in $\gamma$, respectively. The set of every configuration is denoted $\Gamma$.

Let $\mapsto$ be the binary relation between configurations such that $\gamma \mapsto \gamma'$ if the system can reach $\gamma'$ from $\gamma$ by executing an \emph{(atomic) step}. During a step, some processes: (a) receive messages (at most one by incoming channel), (b) do some internal computation, and (c) send some messages (at most one by outgoing channel). An \emph{execution} is a maximal sequence of configuration $e = \gamma_0, \gamma_1, \dots, \gamma_i, \dots$ such that $\forall i > 0$, $\gamma_{i-1} \mapsto \gamma_i$. Configuration $\gamma_0$ is the \emph{initial configuration} of $e$.

For the purpose of complexity analysis, we define a \emph{round} as the smallest execution fragment such that the two following conditions hold: \emph{(i)} all messages that were in transit at the beginning of the round are received (and processed) by the nodes, and \emph{(ii)} all nodes that have no message in transit in their input channels trigger a timeout (and process it).

\subparagraph{Identifier.}

A node $v$ has access to a constant unique identifier $\ID_v$, but can only access its identifier one bit at a time, using the $\Bit_v(i)$ function. 
The function Bit returns the place of the bit at one. Specifically, $ \Bit_v (i) $ returns the position (numbered from right to left) of the $i$th bit to one (from left to right).  Note that since nodes have unique identifiers, they are allowed to execute unique code. For example, suppose node $v$ has identifier $10$ (in decimal notation), or $1010$ (in binary notation). Then, one can implement $\Bit_{v}(i)$ as follows for $v=1010$:

$$
\Bit_{v}(i) := \left\{
\begin{tabular}{lll}
4 &if &i=1\\
2 &if& i=2\\
-1 &if &$i > 2$\\
\end{tabular}
\right.
$$

Since we assume that all identifiers are $O(\log n)$ bits long, the $\Bit_{v}$ function only returns values with $O(\log\log n)$ bits. Also, when executing Function $\Bit_{v}$, the program counter only requires $O(\log\log n)$ values. 
 In turn, this position can be encoded with $O(\log \log n)$ bits when identifiers are encoded using $O(\log n)$ bits, as we assume they are.

\subparagraph{Self-stabilization.} An algorithm $\Alg$ is \emph{self-stabilizing} for some specification $SP$ if there exists a subset of \emph{legitimate configurations} $\Gamma_\ell \subseteq \Gamma$ such that:
\begin{itemize}
    \item {\bf Closure:} $\Gamma_\ell$ is closed, {\emph i.e.}, $\forall \gamma, \gamma' \in \Gamma$ such that $\gamma \mapsto \gamma'$, if $\gamma \in \Gamma_\ell$ then $\gamma' \in \Gamma_\ell$.
    \item {\bf Convergence:} $\Gamma \triangleright \Gamma_\ell$, {\emph i.e.}, for any execution $e = \gamma_0, \gamma_1, \dots, \gamma_i, \dots$ of \Alg starting from an arbitrary initial configuration $\gamma_0 \in \Gamma$, $\exists i \geq 0$ such that $\gamma_i \in \Gamma_\ell$.
    \item {\bf Correctness:} For any execution $e = \gamma_0, \gamma_1, \dots, \gamma_i, \dots$ of \Alg starting from a legitimate configuration $\gamma_0 \in \Gamma_\ell$, $e$ satisfies the specification $SP$.
\end{itemize}

%%%%%%%%%%%%%%%%%%%%%%%%%%%%%%%%%%%%%%%%%%%%%%%%%%%%%%%%%%%%%%%%%%%%%%%%%%%%%%%%%%%%%%%%%%%%%
%%%%%%%%%%%%%%%%%%%%%%%%%%%%%%%%%%%%%%%%%%%%%%%%%%%%%%%%%%%%%%%%%%%%%%%%%%%%%%%%%%%%%%%%%%%%%

\section{Algorithm DAG}

The first layer of our solution consists in providing a symmetry-breaking mechanism. Our approach is to construct a directed acyclic graph (or DAG) based on the unique identifiers of the nodes: hence an edge is to be oriented from the lower identifier to the higher identifier. Of course, neighboring nodes do \emph{not} know the identifier of the other, and should not communicate them directly to each other as it would break the $o(\log n)$ bits constraint on messages.

Our algorithm is presented as Algorithm~\ref{alg:dag}.
For each adjacent link $(v,u)$ of process $v$, $v$ maintains a binary variable $\ord[v](p)$, where $p = \port(v,u)$, as follows. $\ord[v](p)$ represents the orientation of $(v,u)$. More precisely, $\ord[v](p)$ equals 0 if $v$'s identifier is greater than the identifier of $u$ and 1 otherwise. To update variable $\ord[v](p)$, $v$ permanently exchanges its identifier with $u$ in a compact manner using the $\Bit$ function and a counter variable $\cnt_v \in \{1, \dots, \lceil  \log Id_v \rceil\}$ of size $O(\log \log n)$ bits.

When the counter equals to 1, $v$ sends a message to every neighbor $u$ requesting the position of their own first bit to one (see Line~\ref{alg:onlycnt}). A neighbor $u$ answers by a message $<1,B>$, where $B$ is the position of the $1$st bit to one of $u$ (see Line~\ref{alg:repBit}). Now, if $B>\Bit_v(1)$, $u$'s identifier is greater than that of $v$, and if $B<\Bit_v(1)$, then $u$'s identifier is smaller than that of $v$. Finally, if $B=\Bit_v(1)$, the comparison must continue.
When $v$ receives answers from every neighbor about the $1$st position, $v$ increments its counter (see Line~\ref{alg:cnt}. Then, for every neighbor whose status remains unknown (that is, the condition $B=\Bit_v(1)$ was satisfied), $v$ requests the value of the $2$nd bit, and so on until all link orientations are established. This comparison process restarts when the node executes function $Restart(v)$ (see Lines~\ref{alg:restart}-\ref{alg:Rtmp}) that, in particular, sets $\cnt_v$ to 1.

Due to the arbitrary initial configuration, where the link can contain an unknown number of corrupter messages, this process repeats indefinitely. A "Do  forever" process is used to cope with the case of a deadlock due to an absence of messages (see Lines~\ref{alg:DF}-\ref{alg:DF2}). Moreover, messages $<\ell, B>$ that are received by $v$ when its counter $\cnt_v$ is different than $\ell$ are discarded.

In addition to variables $\ord[v](p)$ and $\cnt_v$, process $v$ maintains the following variables. The variable $\tmp_v(p) \in \{0, 1, \bot\}$ is a temporary variable used during the computation of the ordering of the identifiers. When $\tmp_v(p) = \bot$, the orientation of the link $(v,u)$ (if $p = \port(v,u)$) has not been computed yet in this step of the comparison.
Moreover, the variable $\wait_v$ is a set of port numbers used to remember which neighbors have not yet responded during the current step of computation of the ordering.

\begin{algorithm}
%\SetAlgorithmName{}{}

\Fn{Restart(v)}{\label{alg:restart}
 $\cnt_v := 1$; $\wait_v := \ports(v)$;\\\label{alg:Rv}
    \lForall{$p \in \ports(v)$:}{$\tmp_v[p] := \bot$\label{alg:Rtmp}}
}
\Fn{Step(v)}{
    \If{$\wait_v = \emptyset$\label{alg:wait}}{
        \If{$\cnt_v < \lceil\log id_v\rceil$\label{alg:log}}{
            $\cnt_v := \cnt_v+1$;\label{alg:cnt}\\
            $\wait_v := \{ p \in \ports(v) : \tmp_v[p] = \bot \}$ \label{alg:pwait}\\
        }
        \lElse{$Restart(v)$ \label{alg:ERestart}}
    }
    send $<\cnt_v>$ to $\wait_v$ \label{alg:onlycnt};
}

\Upon{\textbf{$<\ell>$ from port $p$}\label{alg:Rell}}
{
	send <$\ell,\Bit_v(\ell)$> to $p$;\label{alg:repBit}\\
	$Step(v)$;
}

\Upon{\textbf{$<\ell,B>$ from port $p$}\label{alg:RlB}}
{
    \If{$(p\in \wait_v) \wedge (\ell=\cnt_v)$\label{alg:waitcnt}}{
        $\wait_v:=\wait_v\setminus\{p\};$\label{alg:delwait}\\
      \lIf{$B>\Bit_v(\ell)$}{$\ord_v[p]:=1;\tmp_v[p]:=1$ \label{alg:sup}}
        \lIf{$B<\Bit_v(\ell) \vee (B=\emptyset)$}{$\ord_v[p]:=0;\tmp_v[p]:=0$}
    }
    $Step(v)$;\label{StepM}
}

\Do{
    \lIf{$\exists p \in \ports(v) : (p = \port(v,u))  \land (Q(u,v) = \emptyset)$ \label{alg:DF}}{
        $Step(v)$\label{alg:DF2}
    }
}

\caption{DAG Algorithm}
\label{alg:dag}
\end{algorithm}

%\paragraph{Correctness of Algorithm DAG.}

We now state the main result about Algorithm~\ref{alg:dag}:

\begin{theorem}
Algorithm~\ref{alg:dag} solves the spanning DAG construction problem in a self-stabilizing manner in $n$-nodes graphs with maximum degree $\Delta$, assuming the message passing model. If the $n$ node identifiers are in $[1,n^c]$ for some $c\geq 1$, then Algorithm~\ref{alg:dag} uses $O(\log \log n+\Delta \log \Delta)$ bits of memory per node and $O(\log \log n)$ bits per message. Moreover, it converges after in $O(k\log n)$ rounds, and the exchange of $O(k\Delta \log n)$ messages, where $k$ (unknown to the algorithm) is the maximum number of (potentially corrupted) messages initially present in each communication link. 
\end{theorem}

 In our algorithm, we use two types of messages:  $<\ell_i>$ and $<\ell_i,B_i>$, where $\ell_i \in \{1, \dots, \lceil \log n\rceil\}$ and $B_i \in \{-1, \dots, \lceil \log n\rceil\}$. In the sequel, we may use $\ell_i$ to denote both $<\ell_i>$ and $<\ell_i,B_i>$.

Let $\lambda_m:\Gamma\times \Msgs \times V \times V \rightarrow \mathbb{N}$ be the following function: 
$$\lambda_m(\gamma,m,u,v)=
\left\{
  \begin{array}{ll}
   1 &\text{if } m=<\ell,B> \wedge (B\neq\Bit_u(\ell)) \\
   0  &\text{otherwise} \\
  \end{array}
\right.
$$
Let $\lambda_Q:\Gamma\times V \times V \rightarrow \mathbb{N}$ be the following function: 
$$\lambda_Q(\gamma,u,v)=\sum_{m\in Q(\gamma,u,v)}\lambda_m(\gamma,m,u,v)$$
Let $\lambda, \Lambda:\Gamma\times V \rightarrow \mathbb{N}$ be the following functions: 
$$\lambda(\gamma,v)=\sum_{u\in N(v)}\lambda_Q(\gamma,u,v)\text{ and }\Lambda(\gamma)=\sum_{v\in V}\lambda(\gamma,v)$$

Finally, we define the set of configurations $\Gamma_{\cal B}$ as follows $\Gamma_{\cal B}=\{\gamma\in \Gamma:\Lambda(\gamma)=0\}$

\begin{lemma}
$\Gamma \triangleright \Gamma_{\cal B}$, and $\Gamma_{\cal B}$ is closed.
\label{lem:CleanLink} 
\end{lemma}

\begin{proof}
When node $v$ receives message $<\ell,B>$ from $v$ in configuration $\gamma$ with $B\neq\Bit_u(\ell)$, $v$ sends in configuration $\gamma'>\gamma$ a message $<cnt_v>$ (see Function $Step(v)$ of Algorithm~\ref{alg:dag}).
If node $u$ receives message $<\ell>$, it sends  $<\ell,\Bit_u(\ell)>$ (see Line \ref{alg:Rell} of Algorithm~\ref{alg:dag}) followed by $<cnt_v>$ (see Function $Step(v)$ of Algorithm~\ref{alg:dag}). So, $u$ does \emph{not} send $<\ell,B>$ with $B\neq\Bit_u(\ell)$.
As a consequence, $\lambda_Q(\gamma,u,v)<\lambda_Q(\gamma',u,v)$ and $\Lambda(\gamma)<\Lambda(\gamma')$.
Note that, we obtain $\lambda_Q(\gamma,u,v)=0$ after at most $k$ receptions of messages through the port $p$ of node $v$ ($p$ being the port number of $v$ leading to $u$).
\end{proof}

Let define $dif(v,u)$ as the first level at which a node $v$ can determine if the identifier of its neighbor $u$ is lower or greater. More formally:
$$dif(v,u) = \min \big\{\ell : \ell \in \{1, \dots, \lceil \log n \rceil\} \land \Bit_v(\ell) \neq \Bit_u(\ell) \big\}$$

We define predicate $Good(\gamma, v, u)$ as follows:
\begin{itemize}
    \item If $id_v < id_u$:
\begin{align}
Good(\gamma,v,u)\equiv & \big[ (\cnt_v(\gamma)\leq dif(u,v)) \land (\tmp_v[\gamma,u]=\bot)\big]\lor \label{eq:}\\    
                     & \big[ (\cnt_v(\gamma)\geq dif(u,v)) \land (\tmp_v[\gamma,u]=1) \land \port(v,u)\not\in \wait_v(\gamma)\big]\label{eq:2}
                    %  \\
                    %  & \big[ (\cnt_v(\gamma)>dif(u,v)) \wedge (\tmp_v[\gamma,u]=1) \wedge \port(v,u)\not\in \wait_v(\gamma)\big]\label{eq:3}
\end{align}
\item  If $id_v > id_u$:
\begin{align}
Good(\gamma,v,u)\equiv & \big[ (\cnt_v(\gamma)\leq dif(u,v)) \land (\tmp_v[\gamma,u]=\bot)\big]\lor \label{eq:3}\\    
                     & \big[ (\cnt_v(\gamma)\geq dif(u,v)) \land (\tmp_v[\gamma,u]=0) \land \port(v,u)\not\in \wait_v(\gamma)\big]\label{eq:4}
                    %  \\
                    %  & \big[ (\cnt_v(\gamma)>dif(u,v)) \wedge (\tmp_v[\gamma,u]=1) \wedge \port(v,u)\not\in \wait_v(\gamma)\big]\label{eq:3}
\end{align}
\end{itemize}

Let $\phi_p(\gamma,u,v):\Gamma\times V \times V \rightarrow \mathbb{N}$ be the following function: 
$$\phi_p(\gamma,u,v)=
\left\{
  \begin{array}{ll}
   0 &\text{if } (id_u>id_v) \wedge ( \ord_v[\gamma, u]=1) \wedge Good(\gamma,v,u)\\
   0 &\text{if }  (id_u<id_v) \wedge (\ord_v[\gamma, u]=0) \wedge Good(\gamma,v,u)\\
   1  &\text{otherwise} \\
  \end{array}
\right.
$$

Let $\phi:\Gamma \times V \rightarrow \mathbb{N}$ and $\Phi : \Gamma \rightarrow \mathbb{N}$ be the following potential functions: 
$$\phi(\gamma,v)=\sum_{u\in N(v)}\phi_p(\gamma,u,v) \qquad \mbox{ and } \qquad \Phi(\gamma)=\sum_{v\in V}\phi(\gamma,v)$$

Finally, we define the set of configurations  $\Gamma_{\tt DAG}$ as the following $\Gamma_{\tt DAG}=\{\gamma\in \Gamma_{\cal B}:\Phi(\gamma)=0\}$

\begin{lemma}
$\Gamma_{\tt DAG}$ is closed.
\label{lem:DAGclose} 
\end{lemma}

\begin{proof}
Let $\gamma \in \Gamma_{\tt DAG}$. Let $\gamma' \in \Gamma$ such that $\gamma \mapsto \gamma'$. 
Without loss of generality, let consider two neighboring processes $u$ and $v$ such that $id_v < id_u$.

Notice that, since $\phi_p(\gamma, u, v) = 0$, $Good(\gamma, v, u)$ holds and $\cnt_v$ is incremented at most once every step or reset to 1 if $v$ executes $Restart(v)$. In the latter case, $v$ also set $\tmp_v[u]$ to $\bot$ so the first line of predicate $Good$ holds. Now, there are three cases:
\begin{itemize}
    \item If $\cnt_v(\gamma) < dif(v,u)$, then $\cnt_v(\gamma') \leq dif(v,u)$ and the value of variable $\tmp_v[u]$ is not changed so $\tmp_v[\gamma,u] = \bot$ and the first line of predicate $Good$ holds.
    \item If $\cnt_v(\gamma) = dif(v,u)$, the value of variable $\tmp_v[u]$ changes only if $v$ executes $Restart(v)$ (see above) or if $v$ receives a message $<\cnt_v,B>$ from $u$. Since $\gamma \in \Gamma_B$ and by definition of $dif(v,u)$, $B > \Bit_v(\cnt_v)$ and $v$ sets $\tmp_v[u]$ to 1 and removes $\port(v,u)$ from $\wait(v)$. Thus, the second line of predicate $Good$ holds.
    \item If $\cnt_v(\gamma) > dif(v,u)$, the value of variable $\tmp_v[u]$ changes only if $Restart(v)$ is executed. Otherwise, the second line of predicate $Good$ remains true.
\end{itemize}
In those three cases, the value of variable $\ord_v[u]$ can only change to the value of $\tmp_v[u] \neq \bot$ so $\ord_v[u]$ can only equal 1.

Thus, $\gamma' \in \Gamma_{\tt DAG}$ and $\Gamma_{\tt DAG}$ is closed.
\end{proof}

\begin{lemma}
$\Gamma_{\cal B} \triangleright \Gamma_{\tt DAG}$.
%\seb{a number of line numbers are missing or are incorrect in the proof of thus lemma, please check.}
\label{lem:DAG} 
\end{lemma}

\begin{proof}
%\seb{Je ne comprends pas cette phrase. Un noeud ne peut pas atteindre une configuration, par ailleurs où sont définies $\gamma_r$ et $\gamma_{dif}$ ??}
To prove $\Gamma_{\cal B} \triangleright \Gamma_{\tt DAG}$, we need to prove that every node $v$ executes at least once the function $Restart(v)$,  because this function initiates again the computation of the DAG for a node $v$. Thanks to the function $Restart(v)$, that puts all the ports number of $v$ in the set $\wait_v$,  the comparison of the place of the first bit is restarted for all the neighbors of $v$. Moreover, since the counter is also restarted at $1$, we compute which neighbors have a greater identity and which one has a smaller identity.   
%in configuration $\gamma_r$, and that starting from configuration $\gamma_r$, the node reaches a configuration $\gamma_{dif}$.

Only the function $Step(v)$ of Algorithm~\ref{alg:dag} can call the function $Restart(v)$ see Line~\ref{alg:ERestart} .
The function $Step(v)$ is used in three cases: 
\begin{itemize}
    \item By the procedure \textbf{Do forever} in order to avoid a deadlock due to an absence of messages in the network (see Line \ref{alg:DF}).
    \item When $v$ receives $<\ell>$ from $u$, to avoid filling queue $Q(v,u)$ with messages $<\ell,B>$, as this would prevent $v$ from updating its own variables (see Line \ref{alg:Rell}).
    \item When $v$ receives $<\ell,B>$ from $u$,  enabling $v$ to update its own variables (see Line \ref{StepM}).
\end{itemize}

%=========

We denote by $dis(\gamma,v,u)$ the distance in the queue of the first message in $Q(\gamma,u,v)$ that decreases $|wait_v|$, if it exists:
$$dis(\gamma,u,v)=\min\{i:m_i\in Q(\gamma,u,v) \wedge (m_i=<\cnt_v(\gamma),B_i>)\}$$
Otherwise, if such a message does not exist, we handle the last message $<\ell_i>$ in $Q(\gamma,v,u)$ with $\ell_i=\cnt_v(\gamma)$. 
Note that, if $dis(\gamma,u,v)=0$ and $Q(\gamma,u,v)$, $v$ must have received all messages in $Q(\gamma,u,v)$ before the reception of messages sent by $u$ after configuration $\gamma$. When $u$ receives $<\ell>$ from $v$, $u$ sends two messages: $<\ell,\Bit_u(\ell)>$, and $<\cnt_u>$ (see Lines~\ref{alg:Rell} and \ref{alg:onlycnt}). As a consequence, the number of messages before $<\cnt_v(\gamma)>$ can be doubled when they are sent in $Q(u,v)$. 
$$I=\min\{i:m_i\in Q(\gamma,v,u) \wedge ( m_i=<\cnt_v(\gamma)>)\}$$
$$dis(\gamma,v,v)=I+|Q(\gamma,u,v)|+2(|Q(\gamma,v,u)|- I) $$

If $dis(\gamma,u,v)=0 \wedge  dis(\gamma,v,v)=0 \wedge Q(\gamma,u,v)\cup Q(\gamma,v,u)\neq\emptyset$. In other words, $\forall m\in Q(\gamma,v,u)$, message $m$ is of type $<\ell,B>$, and $\forall m\in Q(\gamma,u,v)$, message $m$ is of type $<\ell>$. As a consequence, when $v$ receives a message $<\ell>$, thanks to function \textbf{Step(v)} it sends a message $<\cnt_v>$, so the distance to a message $<\ell,B>$ at destination $u$ of $v$ is $q(\gamma,u,v)=|Q(\gamma,v,u)\cup Q(\gamma,u,v)|+1$.
Last, if $Q(\gamma,u,v)\cup Q(\gamma,v,u)=\emptyset$, then $u$ and $v$ can execute \textbf{Do forever} (or only $u$ or only $v$), the maximum distance $2$ is obtained when $u$ and $v$ jointly execute \textbf{Do forever}.  
Let $d(\gamma,u,v):\Gamma\times V \rightarrow \mathbb{N}$ be the following function: 
$$d(\gamma,u,v)=
\left\{
  \begin{array}{ll}
   dis(\gamma,u,v) &\text{if } dis(\gamma,u,v)>0 \\
   dis(\gamma,v,v)&\text{if } dis(\gamma,u,v)=0 \wedge  dis(\gamma,v,v)>0\\
   q(\gamma,v,v)&\text{if } dis(\gamma,u,v)=0 \wedge  dis(\gamma,v,v)=0 \wedge Q(\gamma,u,v)\cup Q(\gamma,v,u)\neq\emptyset \\
3   &\text{if }Q(\gamma,u,v)\cup Q(\gamma,v,u)\neq\emptyset
  \end{array}
\right.
$$
$$D(\gamma,v)=\sum_{u\in wait_v} d(\gamma,u,v)$$
Once $\wait_v$ is empty (see Line \ref{alg:wait}), $v$ increases $\cnt_v$, or executes $Restart(v)$ (see Lines~\ref{alg:log} and \ref{alg:ERestart}). We already proved that in order to reach a configuration where $\phi_p(\gamma,u,v)=0$, $v$ must execute $Restart(v)$. %\seb{(again, a node does not converge to a configuration, and $\gamma_{dif}$ does not seem to be defined.)} 
%then $v$ must have converged to a configuration $\gamma_{dif}$. 
To achieve that we define the function $\epsilon:\Gamma\times V \rightarrow \mathbb{N}$ :
$$\epsilon(\gamma,v)=\left\{
  \begin{array}{ll}
\lceil \log id_c \rceil-\cnt_v(\gamma)+dif(u,v)&\text{if } \cnt_v>dif(u,v)\\
dif(u,v)-\cnt_v&\text{otherwise}
\end{array}
\right.
$$

Let $\alpha:\Gamma\times V \rightarrow \mathbb{N}$ be the following function: 
$$\alpha(\gamma,v)=(\epsilon(\gamma,v),D(\gamma,v))$$
So, when $\alpha(\gamma,v)=0$, for all configuration $\gamma'\geq \gamma$, $\phi(\gamma',v)=0$.
Let us consider a configuration $\gamma'>\gamma$ with $\gamma\in \Gamma_{\cal B}$, we consider different situations:
\begin{itemize}
\item $\wait_v(\gamma)=\emptyset$, and $\forall u\in Q(\gamma,u,v)=\emptyset$, $v$ executes \textbf{Do forever}, and thanks to function \textbf{Step(v)}, $v$ increases or restarts $\cnt_v$, so $\alpha(\gamma',v)<\alpha(\gamma,v)$. 
\item $\wait_v(\gamma)=\emptyset$, and  $\exists u\in Q(\gamma,u,v)\neq\emptyset$, $v$ receives at least one message, so $v$ executes \textbf{Step(v)}, and increases or restarts $\cnt_v$, so $\alpha(\gamma',v)<\alpha(\gamma,v)$. 
\item Now, if $\wait_v(\gamma)\neq\emptyset$, only the messages $<\ell,B>$ received by a port in the set $\wait_v(\gamma)$ can decrease $|\wait_v(\gamma)|$ in order to reach $\wait_v(\gamma')=\emptyset$ (with $\gamma'>\gamma$), and increases or restarts $\cnt_v$ .
\begin{itemize}
    \item If $\forall u\in \wait_v(\gamma)$, we have $Q(\gamma,u,v)\cup Q(\gamma,v,u)=\emptyset$, and function $D(\gamma,v)$ returns $3|\wait_v(\gamma)|$, then $v$ executes \textbf{Do forever}. The execution of \textbf{Do forever} triggers the sending of messages $<\cnt_v>$ through every port $p\in \wait_v(\gamma)$. As a consequence, function $D(\gamma,v)$ returns at most $2|\wait_v(\gamma)|$, so $\alpha(\gamma',v)<\alpha(\gamma,v)$.
    \item If  $\exists u$ such that $Q(\gamma,u,v)\cup Q(\gamma,v,u)\neq\emptyset$,  thanks to function $d(\gamma,u,v)$, we obtain $d(\gamma',u,v)<d(\gamma,u,v)$, so $\alpha(\gamma',v)<\alpha(\gamma,v)$.
\end{itemize}

\end{itemize}

To conclude, we obtain a configuration $\gamma$  where for every $v\in V$, we have $\phi(\gamma,v)=0$,  so $\Phi(\gamma)=0$. Also, for every configuration $\gamma'>\gamma$ we have proved that $\Phi(\gamma')=0$. In other words   $\Gamma_{\cal B} \triangleright \Gamma_{\tt DAG}$.
\end{proof}

\begin{lemma}
Algorithm~\ref{alg:dag} converges in $O(k\log n)$ rounds, and exchanges $O(k\Delta \log n)$ messages before convergence.
\label{lem:DAG} 
\end{lemma}

\begin{proof}
We use the potential function $\Phi$ to compute the number of rounds, remember that $\alpha(\gamma,v)=(\epsilon(\gamma,v),D(\gamma,v))$. Function $\epsilon(\gamma,v)$ is bounded by $\log n$, and Function $D(\gamma,v)$ by $2k$, so in at most $2k\log n$ rounds, each node $v$ computes its set of links directions in the DAG. At each round, $v$ reads $\delta(v)$ messages, so the system converges after $O(k\Delta\log n)$ messages. 
\end{proof}

\section{Vertex Coloring}

We now present our self-stabilizing $(\Delta+1)$ vertex coloring algorithm whose pseudo-code is given in Algorithm~\ref{alg:color}, that is built on top of Algorithm~\ref{alg:dag}. We assume that each node $v$ knows which  neighbor has a lower identifier, and which has a higher one using variables $\ord$ of Algorithm~\ref{alg:dag}. 
Each node $v$ maintains a color variable, denoted by $\clr_v \in \{1, \dots, \delta(v)+1\}$. In addition, $v$ maintains an array with all the known colors of its neighbors denoted $\clr_v[u] \in \{1, \dots, \Delta+1\}$ for each port $u$. To do so, infinitely often, $v$ send its own color to its neighbors by sending a message $<\clr_v>$ (see Lines~\ref{alg:colo:msgNewColor}, \ref{alg:colo:conflict-end}, \ref{alg:colo:msgColo}, and \ref{alg:colo:msgColo2}). If $v$ has a neighbor $u$ whose color is identical (in other words, if $v$ detects a conflict) but identifier is lower, $v$ changes its color for the minimum color not used by its neighbors, by executing function $Conflict(v,u)$ (see Lines~\ref{alg:colo:conflict-begin}-\ref{alg:colo:conflict-end}).

Note that, since the maximum degree of the graph is $\Delta$, the size of $\clr_v$ is bounded by $\Delta \log \Delta$.

\begin{algorithm}
%\SetAlgorithmName{}{}

\Fn(\tcc*[f]{$\clr_v[u] = \clr_v$}){Conflict(v,u)}{\label{alg:colo:conflict-begin}
    \uIf{$(\forall w \in \ports(v) : (\clr_v[w] \neq \bot) \land \big((\clr_v[w] = \clr_v) \Rightarrow (\ord_v[w] = 1) \big)$\label{alg:colo:ifMinColor}}{
        $\clr_v := \min \{1, \dots, \delta(v)\} \setminus \{\clr_v[w] : w \in \ports(v)\}$;\label{alg:colo:minColor}\\
        \lForall{$w \in \ports(v):$}{
            % $\clr_v[w] := \bot$;\\
            send $<\clr_v>$ to $w$\label{alg:colo:msgNewColor}
        }
    }
    \lElse{send $<\clr_v>$ to $u$}\label{alg:colo:conflict-end}
}

\Upon{\textbf{$<c>$ from port $u$} }{
    $\clr_v[u] := c$;\label{alg:colo:localUpdate}\\
    \lIf{$\clr_v[u] = \clr_v$}{$Conflict(v,u)$}
    \lElse{send $<\clr_v>$ to $u$}\label{alg:colo:msgColo}
}

\Do{
    \Forall{$u \in \ports(v) : Q(u,v) = \emptyset$}{
        \lIf{$\clr_v[u] = \clr_v$}{
        $Conflict(v,u)$
        }
        \lElse{
            send $<\clr_v>$ to $u$\label{alg:colo:msgColo2}
        }
    }
}

\caption{Vertex coloring}
\label{alg:color}
\end{algorithm}

%We now state the main theorem related to Algorithm~\ref{alg:color}:

\begin{theorem}
Algorithm~\ref{alg:color} solves the vertex coloring problem in a self-stabilizing manner in $n$-nodes graph with maximum degree $\Delta$, assuming the message passing model and a spanning DAG. It uses $O(\Delta \log \Delta)$ bits of memory per node, $O(\log \Delta)$ bits per message. Moreover, it converges after $O(n)$ rounds and exchanges $O(\Delta~k~n)$ messages.\end{theorem}

In the following proof, we assume that the oriented graph described by variables $\ord$ is a spanning DAG.

\begin{lemma}\label{lem:SendColor}
Every time process $v$ executes a step, it sends $<\clr_v>$ to every neighbor $u \in N(v)$.
\end{lemma}

\begin{proof}
Each time $v$ executes a step, for every neighbor $u \in N(v)$, there are two cases:
\begin{itemize}
    \item $v$ receives a message from $u$, so $v$ executes $Conflict(v,u)$, or sends back $<\clr_v>$ to $u$.
    \item $v$ executes {\bf Do forever}, so it calls $Conflict(v,u)$, or sends back $<\clr_v>$ to $u$.
\end{itemize}
If $v$ calls $Conflict(v)$, either $v$ sends back $<\clr_v>$ to $u$, or $v$ changes its color and sends the new one to every neighbor.
\end{proof}

Let $\alpha_1 : \Gamma \times V \rightarrow \mathbb{N}$ and $\alpha_2 : \Gamma \times V \rightarrow \mathbb{N}$ be the following functions:
$$\alpha_1(\gamma,v) = \left| \{u \in N(v) : \clr_u = \clr_v \}\right| \qquad \text{and} \qquad \alpha_2(\gamma,v) = \left| \{u \in N(v) : \clr_v[u] \neq \clr_u \}\right|$$

Let $\mu : \Gamma \times V \times \{1, \dots, \Delta\} \rightarrow \mathbb{N}$ be the following function:
$$\mu(\gamma,v,c) = \left\{
\begin{array}{ll}
     1 & \text{if } \clr_v \neq c \\
     0 & \text{otherwise}
\end{array}
\right.$$

Let $\alpha_3 : \Gamma \times V \rightarrow \mathbb{N}$ be the following function:
$$\alpha_3(\gamma,v) = \sum_{u \in N(v)} \left(\sum_{<c> \in Q(u,v)} \mu(\gamma,u,c) \right)$$

% Let $\alpha : \Gamma \times V \rightarrow \mathbb{N}$ be the following function:
% $$\alpha(\gamma,v) = \alpha_1(\gamma,v) + \alpha_2(\gamma,v) + \alpha_3(\gamma,v)$$

Let $A : \Gamma \rightarrow \mathbb{N}$ be the following potential function:
$$A(\gamma) = \sum_{v \in V}  \big( \alpha_1(\gamma,v) + \alpha_2(\gamma,v) + \alpha_3(\gamma,v) \big)$$

We define
$\Gamma_\alpha = \{ \gamma \in \Gamma : A(\gamma) = 0 \}$.

\begin{lemma}
$\Gamma \triangleright \Gamma_{\alpha}$.
\label{lem:NoConflict:convergence} 
\end{lemma}

\begin{proof}
Let $d_{s}(v)$ be the minimum distance from process $v$ to a source of the DAG ({\em i.e.}, a process with only outgoing edges).
First, we show that in finite time, every process $v$ stops changing its color by induction on $d_s(v)$.
\begin{itemize}
    \item {\bf Base case:} If $d_s(v) = 0$, then $v$ is a source of the DAG and $\forall u \in N(v)$, $\ord_v[u] = 0$. 
    Thus, $v$ never changes its color $\clr_v$ (see Line~\ref{alg:colo:ifMinColor}).
    
    \item {\bf Induction step:} If $d_s(v) = x+1$, then $\forall u \in N(v)$, either $\ord_v[u] = 0$, or $d_s(u) < d_s(v) = x+1$.
    By induction hypothesis, in finite time, $\forall u \in N(v)$ such that $d_s(u) \leq x$, the value of $\clr_u$ eventually stops changing. 
    Then, by Lemma~\ref{lem:SendColor}, $u$ sends $<\clr_u>$ to $v$ infinitely often, so eventually $\clr_v[u] = \clr_u$. 
    Now, either $v$ never changes its color, or it gets a new color different from $\clr_u$. 
    In the second case, $v$ is no longer in conflict with any neighbor $u \in N(v)$ such that $d_s(u) < d_s(v)$. 
    Node $v$ can only be in conflict with neighbors $w \in N(v)$ such that $\ord_v[w] = 0$, and thus $v$ does not change its color anymore (see Line~\ref{alg:colo:ifMinColor}). 
    Hence, $v$ can change its color at most once.
\end{itemize}

Once every process stops changing its color forever, a process $v$ cannot send a color different from its own.
Thus, there exists a configuration $\gamma_a$ such that, in every subsequent configuration $\gamma$, for every process $v$, $v$ does not change its color and $\alpha_3(\gamma,v) = 0$.
Moreover, in finite time, processes update their local knowledge about their neighbors' colors (see Lemma~\ref{lem:SendColor} and Line~\ref{alg:colo:localUpdate}). 
Thus, there is a configuration $\gamma_b$ after $\gamma_a$ such that, in every configuration $\gamma$ after $\gamma_b$, for every process $v$, $\alpha_2(\gamma,v) = 0$.

Finally, we show that in every configuration after $\gamma_b$, there is no color conflict. 
Assume by contradiction that in some configuration $\gamma$ after $\gamma_b$, there is a color conflict.
Let consider the process $v$ in conflict that has the lowest ID among all such processes. 
For every $u \in N(v)$ such that $\clr_v = \clr_u$, $\clr_v[u] = \clr_u = \clr_v$, and $\clr_u[v] = \clr_v = \clr_u$. 
Moreover, by assumption, $\ord_v[u] = 1$ and $\ord_u[v] = 0$. Thus, in finite time, $v$ calls the function $Conflict$, and changes its color (see Line~\ref{alg:colo:ifMinColor}), a contradiction.
So, in every configuration $\gamma$ after $\gamma_b$, and for every process $v\in V$, $\alpha_1(\gamma,v) =0$.
Hence, $A(\gamma) = 0$, and $\Gamma \triangleright \Gamma_\alpha$.
\end{proof}

\begin{lemma}
$\Gamma_{\alpha}$ is closed.
\label{lem:NoConflict:closure} 
\end{lemma}

\begin{proof}
Let $\gamma \mapsto \gamma'$ such that $\gamma \in \Gamma_{\alpha}$. 
Let $v \in V$ and $u \in N(v)$.

First, in $\gamma$, $\clr_v[u] = \clr_u \neq \clr_v$. Moreover, every message in $Q(u,v)$ contains color $\clr_u \neq \clr_v$. 
Thus, $v$ does not change its color during step $\gamma \mapsto \gamma'$. 
Hence,  $\alpha_1(\gamma',v) = 0$.

Since $u$ and $v$ do not change their color between $\gamma$ and $\gamma'$, and since every message in $Q(u,v)$ in $\gamma$ equals $\clr_u$, then the value of $\clr_v[u]$ does not change, and remains equal to $\clr_u$ in $\gamma'$. 
Hence, $\alpha_2(\gamma',v) = 0$.

Finally, $u$ does not change its color between $\gamma$ and $\gamma'$, and every message sent by $u$ during step $\gamma \mapsto \gamma'$ equals $<\clr_u>$ such as every other message in $Q(u,v)$ in $\gamma$ that $u$ did not treat during step $\gamma \mapsto \gamma'$.
% \seb{Je ne comprends pas la phrase précédente.}
Thus, $\alpha_3(\gamma',v) = 0$
\end{proof}

\begin{lemma}
Algorithm~\ref{alg:color} converges in $O(n)$ rounds and exchanges $O(\Delta~k~n)$ messages before convergence.
\label{lem:colormsg} 
\end{lemma}

\begin{proof}
After one round, the system is purged of all possibly erroneous messages contained in the initial configuration.  Moreover, by Lemma~\ref{lem:SendColor}, at each round a process $v$ sends its color $<\clr_v>$ to all its neighbors, neighbors that will receive this message during the next round. 

Now, consider a process $v$ that changes its color during some round. The last time $v$ changes its color during the round, it knows the real colors of its neighbors, {\em i.e.}, $\forall u \in N(v)$, $\clr_v[u] = \clr_u$. So by changing its color, $v$ solves its color conflicts with its neighbors having a greater identity and does not create any new conflict with them. Since a process can change its color only if it is in conflict with a node of greater identity, $v$ will not change its color anymore after this round.
Hence, in the worst case, processes stop changing their colors after $O(n)$ rounds and there is no color conflict anymore. 

In the worst case, at each round, a process replies to every message sent by its neighbors during the previous round. Thus, $O(\Delta~k~n)$ messages are exchanged.

% The worst case for the time convergence appears when the DAG forms a chain, and all nodes are in conflict (that is, there are $n$ conflicts). Now, a node changes its color when it is in conflict with a node with higher identity. So in the first round after the purge, $n-1$ nodes can change their colors and send their new colors. This changes can result in $n-2$ conflicts. This process is repeated until no conflict occurs, that is $O(n)$ rounds.  
% During this process, $O(kn)$ messages are exchanged.
\end{proof}

To avoid starvation of one of the algorithms, we compose them using a {\em fair composition}~\cite{D00b}, {\em i.e.}, each process alternatively executes a step of Algorithm~\ref{alg:dag} and a step of Algorithm~\ref{alg:color}. As demonstrated in previous work~\cite{D00b}, fair composition preserves the property of self-stabilization.

\section{Concluding Remarks}

We presented the first deterministic self-stabilizing solution in asynchronous message passing networks with unknown capacity links that only requires sublogarithmic (in $n$) memory and message size, when $\Delta$ is itself sublogarithmic.

Our approach is constructive and modular. In particular, our \textbf{DAG} algorithm layer solves a fundamental difficulty in many settings with respect to self-stabilization: avoidance of cyclic behaviors. We believe this can be a valuable asset when solving other problems in the same setting. 

While our \textbf{Vertex coloring} algorithm layer does not guarantee a locally minimal coloring (if an initial configuration is a coloring that is not locally minimal, no conflict is found and thus recoloring occurs), a simple modification of the protocol permits to achieve this result: a node $v$ that does not see a conflict but would like a smaller color asks its higher identity neighbors the authorization to take a new color $c$, any such neighbor $u$ grants the authorization unless it has color $c$ (but $v$ didn't know it), or it itself wants to take a new color and is waiting for authorization. This mechanism does not create new conflicts, and any authorization chain length is bounded by the height of our constructed DAG. As a result, a minimal coloring is obtained, which can then be used for solving the maximal independent set problem: as it provides a locally minimal coloring, any node with color $0$ can be seen as a member of the maximal independent set.

We believe the \textbf{DAG} and \textbf{Vertex Coloring} can prove useful for other local tasks, such as minimal dominating set, link coloring, or maximal matching. It would also be interesting to investigate their utility for solving global tasks (such as tree construction or leader election), and whether resource efficiency remains in this setting.

One property we retained about communication links is their FIFO behavior. Relieving this hypothesis in the context of unknown capacity links is likely to generate many impossibility results, that are left for future work.

%\newpage
%%%%%%%%%%%%%%%%%%%

\bibliographystyle{plainurl}
\bibliography{biblio}%,BibPointeurs}

\newpage

% \appendix

% \section{Appendix : Omitted proofs}

\end{document}